\documentclass[11pt]{article}

\usepackage[T1]{fontenc}
\usepackage[utf8]{inputenc}
\usepackage{lmodern}
\usepackage{amsmath,amssymb,amsthm,mathtools,bm}
\usepackage{booktabs}
\usepackage{array}
\usepackage{microtype}
\usepackage[a4paper,margin=25mm]{geometry}
\usepackage{graphicx}
\usepackage[hidelinks]{hyperref}
\hypersetup{pdfauthor={Marina O. Gordeichuk and Oleg M. Kiselev},
 pdftitle={Measurement-dependent reduction of open dynamics in the quantum phase estimation algorithm: the visible subspace and the damped Fejer kernel}}

\newtheorem{theorem}{Theorem}[section]
\newtheorem{proposition}[theorem]{Proposition}
\newtheorem{corollary}[theorem]{Corollary}

\theoremstyle{definition}
\newtheorem{definition}[theorem]{Definition}
\newtheorem{example}[theorem]{Example}
\theoremstyle{remark}
\newtheorem{remark}[theorem]{Remark}

\newcommand{\Tr}{\operatorname{Tr}}
\newcommand{\Span}{\operatorname{span}}
\newcommand{\id}{\mathrm{id}}
\newcommand{\cL}{\mathcal L}
\newcommand{\cK}{\mathcal K}
\newcommand{\cE}{\mathcal E}
\newcommand{\cH}{\mathcal H}

\newcommand{\ket}[1]{\lvert #1\rangle}
\newcommand{\bra}[1]{\langle #1\rvert}
\newcommand{\proj}[1]{\lvert #1\rangle\!\langle #1\rvert}
\newcommand{\norm}[1]{\left\lVert #1\right\rVert}

\begin{document}

\title{Measurement-dependent reduction of open dynamics in the quantum phase estimation algorithm: the visible subspace and the damped Fej\'er kernel}

\author{Marina~O.~Gordeichuk\thanks{Corresponding author: \href{mailto:m.lisnichenko@innopolis.university}{m.lisnichenko@innopolis.university}.
}
\and Oleg~M.~Kiselev\thanks{Email: \href{mailto:o.kiselev@innopolis.ru}{o.kiselev@innopolis.ru}.}}

\date{}
\maketitle
\begin{center}\small
Innopolis University, 1 Universitetskaya Street,\\
Innopolis 420500, Republic of Tatarstan, Russia
\end{center}

\begin{abstract}
For the quantum phase estimation (QPE) algorithm we study the minimal space of
observables that is sufficient to predict the entire output distribution under
open dynamics of the counting register. Building on the general Krylov theory
of reduction, we compute the structure of this space explicitly for the full
Fourier measurement. The outcome probabilities are determined by cyclic sums
of the state matrix. If the adjoint generator is diagonal on the matrix units,
the spectral classes inside the residue classes of the index difference modulo
the register dimension $M$ give a constructive basis, an exact dimension
formula, and attainable bounds: the dimension of the visible space lies
between $M$ and $M^2-M+1$. The class includes independent qubit dephasing, and
an analytically solvable reference model of Gorini--Kossakowski--Lindblad--Sudarshan
(GKLS) type yields a damped Fej\'er kernel; two random-unitary channels with
the same Ramsey contrast show that higher harmonics are distinguishable in the
QPE distribution. For approximate reduction we use a finite-time Gramian and an
a posteriori residual estimate, with a quadratic error order for a reducing
unperturbed subspace and a compatible initial state. Numerical calculations
check the consistency of the representations and the behaviour of the error
under non-Toeplitz perturbations.
\end{abstract}
\noindent\textbf{Keywords:} quantum phase estimation; quantum Markov semigroups;
observable subspace; visibility Gramian; Krylov reduction; Fej\'er kernel;
open quantum systems.

\section{Introduction}

The quantum phase estimation (QPE) algorithm is one of the principal spectral
subroutines of quantum computing. In the canonical scheme a counting register
of $n$ qubits resolves the phase on a grid of $M=2^n$ points by means of
controlled powers of a unitary operator and the inverse quantum Fourier
transform (QFT) \cite{Kitaev1997,NielsenChuang}. The ideal grid step $M^{-1}$
sets the resolution of the ideal phase grid. Long controlled evolutions,
an inexact implementation of the Hamiltonian, heating of vibrational modes,
dephasing, relaxation and readout errors act at different stages of the
protocol and change the measured distribution in different ways.

In \cite{GordeychukKiselev2026} a stability condition of the type
$\epsilon M\ll1$ was obtained for an operator perturbation of size
$\epsilon$. For open dynamics one has to account separately for the coherent
phase shift, the continuous dissipative generator and the classical readout
channel. Long time branches make QPE more sensitive to these mechanisms.

Here we construct a reduction of observables. The main
question is formulated as follows: which modes of the full Liouvillian survive
a fixed QPE measurement? For teleportation an analogous operational idea leads
to modes that are visible through the populations of the Bell states. In QPE
the geometry is different: the final measurement after the inverse Fourier
transform depends only on cyclic sums of entries of the counting-register
density matrix. Dynamics can make an initially invisible state difference
visible later. The invariant closure of the measurement effects determines
the outcome probabilities at all times.

A single-step Ramsey interference experiment probes the first harmonic of
coherence transfer $g_1$: the contrast determines $|g_1|$, and phase-sensitive
measurements determine its argument.
The QPE distribution depends on cyclic combinations of several
harmonics $g_1,\ldots,g_{M-1}$.
Proposition~\ref{prop:ramsey-separation-ru} below gives two physical
random-unitary channels with the same Ramsey contrast and different QPE
distributions. Thus the multiscale structure of the controlled powers of $U$
makes the visible reduction natural precisely for QPE.

The main contribution is the explicit computation of the minimal space of
observables of the full QPE positive operator-valued measure (POVM) for generators that are diagonal on the
matrix units: cyclic spectral classes give its basis and its exact dimension.
The general Krylov construction and the minimality criterion are used as a
known foundation \cite{Jamiolkowski1983,Grigoletto2025}; when diagonality is
violated, a posteriori control of the approximate reduction is applied.
The main results are as follows. We obtain an exact QPE reduction to cyclic
sums of diagonals and prove that channels with the same Ramsey contrast can be
told apart by the QPE distribution. For generators that are diagonal on the
matrix units we prove the spectral decomposition of the QPE-visible space and
the exact formula for its dimension; the Toeplitz law and the
Gorini--Kossakowski--Lindblad--Sudarshan (GKLS) generator of the damped
Fej\'er kernel follow as a special case. We construct a finite-time visibility
Gramian whose spectral truncation is optimal in the integral norm, and we prove
a quadratic law for the total weight of new modes under a small perturbation of
the generator. Finally, we prove an a posteriori stability estimate in terms of
the invariance residual, together with a special quadratic estimate, and
compare the full propagator, the reduced model and the analytic kernel
numerically, including a non-Toeplitz GKLS perturbation at $M=8,16,32$.

\subsection{Relation to previous work and novelty}

The algebraic basis of the construction is related to the observability of
linear systems and to stroboscopic tomography of open quantum dynamics
\cite{Jamiolkowski1983,Czerwinski2017}. The general theory of exact reduction
of continuous-time open quantum dynamics through Krylov spaces of initial
states and observables was developed in \cite{Grigoletto2025}; extending these
spaces to operator algebras ensures there a physically admissible reduced
generator. The space \eqref{eq:visible-space-ru} is the observable part of
that construction for a fixed POVM. Here its explicit computation for the full
vector of QPE outcomes is carried out.

Discrete quantum channels and controlled time-dependent generators are covered
by the general reduction schemes of
\cite{GrigolettoTicozzi2025,GrigolettoViolaTicozzi2025}. The approach of
\cite{FanTureci2026} builds an effective equation from the finite time
resolution of the measurement channel. In the present work the measurement is
fixed to be the full QPE POVM, and the reduction preserves exactly the whole
vector of its probabilities;
Theorem~\ref{thm:diagonal-generator-dimension-ru} computes the corresponding
minimal linear model in closed form.

Approximate reduction onto the central subspace of an unperturbed generator
was considered in \cite{GrigolettoSarlette2026}. In particular, Theorem~2(ii)
of the cited version gives an error estimate for the projected dynamics that is quadratic
in the perturbation for a compatible initial state. In the present work the
reducing subspace is chosen according to the QPE observables and may include
decaying modes; the estimate is applied directly to the outcome probabilities.
The quadratic order itself has a general projection--perturbation mechanism.

Krylov approximations of the matrix exponential and their convergence were
studied in \cite{Saad1992,HochbruckLubich1997}, and projection-based reduction
of linear dynamical systems was systematized in \cite{Bai2002}. Observability
Gramians, principal components of trajectories and spectral truncation underlie
the classical theory of reduction \cite{Moore1981,Antoulas2005}. Residual
criteria and a posteriori estimates for the action of the matrix exponential
were constructed in \cite{Botchev2013,JiaLv2015}.
Theorem~\ref{thm:residual-stability-ru} transfers this principle to the space
of quantum observables and directly bounds each POVM probability
simultaneously for all initial states of a given purity.

The authors' previous work \cite{GordeychukKiselev2026} established stability
estimates for QPE with respect to an operator perturbation of the input unitary
operator. The present work moves the question to the class of open GKLS
dynamics \cite{Gorini1976,Lindblad1976,BreuerPetruccione} and singles out the
observable part of the generator. Works on noisy QPE and sequential quantum
channels \cite{OBrien2019,Jin2025,Patel2026} study the stability and
modifications of the protocol. Here the full QPE POVM is retained, and its
output statistics factorize through the minimal space of observables.
Phase noise in Fourier-based QPE and its influence on the choice of the
computational resource were studied in \cite{ChapeauBelin2020}. In
\cite{Smith2025} a Toeplitz matrix formed by the Fourier coefficients of the
loss function determines the input state with minimal risk. Here the Toeplitz
structure arises in the coherence model, while the cyclic sums are determined
by the POVM itself and apply to an arbitrary state matrix.

For the full QPE POVM the general Krylov construction
\eqref{eq:visible-space-ru} is computed in closed form. The residue classes
$j-\ell\bmod M$ that arise from the discrete Fourier transform of the outcomes
give an explicit basis $B_{r,\zeta}$, an exact dimension formula, sharp bounds
and a Toeplitz specialization. The finite-time Gramian provides a stable
spectral scale, and its quadratic tail links the exact rank with approximate
reduction.

\section{QPE as a fixed measurement projection}

Let $M=2^n$, and let $\cH_C\cong\mathbb C^M$ be the space of the counting
register. The work register stores the state of the system to which $U$ is
applied; the counting register of $n$ qubits encodes the phase estimate by the
outcomes $k/M$. The vector $\ket j$ denotes the $j$-th computational basis
vector, labelled by the binary representation of the number $j$.
Let $\cH_S$ be the space of the work register and
\[
 U\ket{\psi_\omega}=e^{2\pi i\omega}\ket{\psi_\omega},
 \qquad \omega\in[0,1).
\]
After the Hadamard layer and the controlled powers of $U$ the ideal state of
the counting register has the form
\begin{equation}
 \ket{\phi_\omega}
 =\frac{1}{\sqrt M}\sum_{j=0}^{M-1}e^{2\pi i j\omega}\ket j.
 \label{eq:clock-state-ru}
\end{equation}
We define the Fourier transform by the convention
\[
 F_M\ket k=\frac{1}{\sqrt M}\sum_{j=0}^{M-1}e^{2\pi i jk/M}\ket j.
\]
If the inverse Fourier transform is moved onto the measurement operators, one
obtains the effects
\begin{equation}
 E_k:=F_M\proj{k}F_M^\dagger,
 \qquad k=0,\ldots,M-1.
 \label{eq:qpe-effects-ru}
\end{equation}
The effects $E_k\ge0$, $\sum_kE_k=I$ form a POVM; here they are
orthogonal rank-one projectors. For a state $\rho_C$
immediately before the inverse QFT the outcome probability is
\begin{equation}
 p_k=\Tr(\rho_CE_k).
 \label{eq:qpe-projection-ru}
\end{equation}

Substituting the state \eqref{eq:clock-state-ru} gives the Fej\'er kernel
\begin{equation}
 p_k^{(0)}(\omega)
 =\frac{1}{M^2}
 \frac{\sin^2\!\bigl(\pi M(\omega-k/M)\bigr)}
 {\sin^2\!\bigl(\pi(\omega-k/M)\bigr)},
 \label{eq:ideal-fejer-ru}
\end{equation}
extended by the value $p_k^{(0)}=1$ when the phase coincides exactly with a
grid node. Formula \eqref{eq:qpe-projection-ru} is the operational starting
point and immediately separates three things: the open evolution before the
inverse QFT, the coherent detuning that shifts the phase factors, and the
classical stochastic readout channel that acts after the quantum measurement.

\section{Minimal space of visible Liouvillian modes}

We first formulate the construction for an arbitrary finite-dimensional
protocol. Let $\{E_k\}_{k=1}^m$ be a POVM in $\cH$, let
$\cE_t=e^{t\cL}$ be a quantum Markov semigroup, and let
$\cE_t^\dagger=e^{tG}$ be its adjoint evolution in the Heisenberg picture,
where $G:=\cL^\dagger$ is the generator of observables. The notation $G$ is
used in all later sections. The space of operators $\mathcal B(\cH)$ is
considered over $\mathbb C$ with the Hilbert--Schmidt inner product
$\langle X,Y\rangle_{\rm HS}=\Tr(X^\dagger Y)$. All restrictions of generators
and orthogonal projections below refer to this finite-dimensional structure.

\begin{definition}[Space of visible observables]
The minimal invariant space visible to the POVM
$E=(E_1,\ldots,E_m)$ is
\begin{equation}
 \cK_E
 :=\Span_{\mathbb C}
 \bigl\{G^r(E_k):k=1,\ldots,m,\ r\ge0\bigr\}.
 \label{eq:visible-space-ru}
\end{equation}
For QPE with the effects \eqref{eq:qpe-effects-ru} the notation
$\cK_{\rm QPE}$ is used.
\end{definition}

\begin{theorem}[Exact reduction to the visible space]
\label{thm:visible-reduction-ru}
For any initial state $\rho_0$ each probability
\begin{equation}
 p_k(t)=\Tr\!\left[\rho_0e^{tG}(E_k)\right]
 \label{eq:heisenberg-probability-ru}
\end{equation}
is completely determined by the restriction
\begin{equation}
 A_E:=G|_{\cK_E}.
\end{equation}
Moreover, $p_k(t)$ is a finite linear combination of the functions
\begin{equation}
 t^re^{\lambda t},
 \qquad \lambda\in\sigma(A_E),
 \qquad 0\le r<\nu(\lambda),
\end{equation}
where $\nu(\lambda)$ is the size of the largest Jordan block for $\lambda$.
\end{theorem}

\begin{proof}
By construction $G(\cK_E)\subseteq\cK_E$: the image of each generating element
$G^rE_k$ is the next generating element. In finite dimension the sequence of
Krylov spans stabilizes, so the restriction $A_E$ is a finite-dimensional
operator. From $E_k\in\cK_E$ it follows that
\[
 G^rE_k=A_E^rE_k\qquad (r\ge0).
\]
Comparing the power series of the exponentials gives
$e^{tG}E_k=e^{tA_E}E_k$. Applying the linear functional
$X\mapsto\Tr(\rho_0X)$ proves the first statement.

Choose a Jordan basis for $A_E$. The exponential of a block
$J_\lambda=\lambda I+N$ of size $s$ has the form
\[
 e^{tJ_\lambda}=e^{\lambda t}
 \sum_{r=0}^{s-1}\frac{t^rN^r}{r!}.
\]
After applying the linear functional one is left with finite linear
combinations of $t^re^{\lambda t}$ with $r<\nu(\lambda)$, which proves the
second statement.
\end{proof}

The following standard consequence of observability is given with a proof in
order to specify the class of models being compared
\cite{Jamiolkowski1983,Grigoletto2025}.
\begin{proposition}[Universal minimality]
\label{thm:universal-minimality-ru}
Let $\mathcal W\subseteq\mathcal B(\cH)$ be a linear space with the following
property: for each outcome $k$ and time $t\ge0$ there is an operator
$Y_k(t)\in\mathcal W$ such that
\begin{equation}
 \Tr\!\left[\rho Y_k(t)\right]
 =\Tr\!\left[\rho e^{tG}(E_k)\right]
 \label{eq:universal-representation-ru}
\end{equation}
for all density matrices $\rho$. Then $\cK_E\subseteq\mathcal W$.
Consequently, $\cK_E$ is the unique smallest linear carrier of an exact
representation of all outcome probabilities that is uniform in the initial
states and in time.
\end{proposition}

\begin{proof}
Density matrices linearly span the space of operators and separate its
elements with respect to the trace pairing. Therefore
\eqref{eq:universal-representation-ru} implies
\[
 Y_k(t)=e^{tG}(E_k)\in\mathcal W.
\]
In finite dimension the space $\mathcal W$ is closed. Successive
differentiation of the trajectory at $t=0$ gives
\[
 G^r(E_k)
 =\left.\frac{d^r}{dt^r}e^{tG}(E_k)\right|_{t=0}
 \in\mathcal W,
 \qquad r\ge0.
\]
Hence all generating elements of \eqref{eq:visible-space-ru} belong to
$\mathcal W$. The converse sufficiency follows from
Theorem~\ref{thm:visible-reduction-ru}.
\end{proof}

\begin{remark}[Restricted family of initial states]
\label{rem:restricted-state-family-ru}
The universal minimality in Proposition~\ref{thm:universal-minimality-ru}
refers to the whole set of density matrices. For a given family $\mathfrak S$
the statistics determine the observables only up to the annihilator
\[
 \mathfrak S^\perp
 :=\{X:\Tr(\rho X)=0\ \text{for all }\rho\in\mathfrak S\}.
\]
The quotient $\cK_E/(\cK_E\cap\mathfrak S^\perp)$ identifies observables
with identical expectations on $\mathfrak S$. All dimension formulas in this work refer
to the universal model that preserves the statistics for an arbitrary initial
state.
\end{remark}

\begin{remark}[Linear observable model]
\label{rem:linear-versus-physical-reduction-ru}
The restriction $A_E$ is an exact linear model of the evolution of the chosen
observables. All dimensions $\dim\cK_E$ below refer to this linear
observable model. The explicit construction \eqref{eq:brownian-gkls-ru}
specifies a GKLS generator on the counting register and gives an exact
realization of the computed observable dynamics.
\end{remark}

\subsection{Finite-time Gramian and effective visibility}

The exact dimension of the visible space is sensitive to the splitting of
multiple eigenvalues. To describe the strength of the visible modes
quantitatively we fix a finite time horizon $T>0$, set
\begin{equation}
 x_k(t):=e^{tG}E_k,
 \label{eq:visible-trajectories-ru}
\end{equation}
and define a positive operator $\mathfrak G_T$ in the Hilbert space
$\mathcal B(\cH)$ by
\begin{equation}
 \mathfrak G_T X
 :=\sum_{k=1}^{m}\int_0^T
 x_k(t)\,\langle x_k(t),X\rangle_{\rm HS}\,dt.
 \label{eq:finite-time-gramian-ru}
\end{equation}
Formula \eqref{eq:finite-time-gramian-ru} is a finite-time, measurement-dependent
analogue of the observability Gramian \cite{Moore1981,Antoulas2005}: it
contains the complete set of effects of the chosen POVM. The following result
links its rank with universal minimality and gives the optimal truncation
directly for the probability trajectories.

\begin{theorem}[QPE-visibility Gramian]
\label{thm:finite-time-gramian-ru}
Let
$\sigma_1(T)\ge\cdots\ge\sigma_D(T)\ge0$ be the eigenvalues of
$\mathfrak G_T$, where $D=\dim\mathcal B(\cH)$. Then (i)~$\mathfrak G_T$ is
self-adjoint and positive, and
\begin{equation}
  \operatorname{Ran}\mathfrak G_T=\cK_E,
  \qquad
  \operatorname{rank}\mathfrak G_T=\dim\cK_E;
  \label{eq:gramian-range-ru}
\end{equation}
(ii)~among all orthogonal projectors $P$ of rank $q$ the minimum of the
functional
\begin{equation}
  \mathcal J_T(P)
  :=\sum_{k=1}^{m}\int_0^T
  \|(I-P)x_k(t)\|_{\rm HS}^2\,dt
  \label{eq:gramian-projection-cost-ru}
\end{equation}
equals $\sum_{j>q}\sigma_j(T)$ and is attained at the spectral projector
$P_q$ onto the span of $q$ orthonormal eigenvectors corresponding to
$\sigma_1(T),\ldots,\sigma_q(T)$; and (iii)~if
\begin{equation}
  \widehat p_{k,q}(t)
  :=\Tr\!\left[\rho_0P_qx_k(t)\right],
  \label{eq:gramian-projected-probability-ru}
\end{equation}
then, with the convention $\sigma_{D+1}(T)=0$,
\begin{equation}
  \sum_{k=1}^{m}\int_0^T
  |p_k(t)-\widehat p_{k,q}(t)|^2\,dt
  \le \sigma_{q+1}(T)\|\rho_0\|_{\rm HS}^2.
  \label{eq:gramian-probability-bound-ru}
\end{equation}
\end{theorem}

\begin{proof}
From \eqref{eq:finite-time-gramian-ru} it follows that
\begin{equation}
 \langle X,\mathfrak G_TX\rangle_{\rm HS}
 =\sum_{k=1}^{m}\int_0^T
 |\langle x_k(t),X\rangle_{\rm HS}|^2\,dt,
 \label{eq:gramian-quadratic-form-ru}
\end{equation}
which proves self-adjointness and positivity. Denote by
$\mathcal V_T$ the linear span of all $x_k(t)$ for $0\le t\le T$.
Each trajectory belongs to $\cK_E$, hence
$\mathcal V_T\subseteq\cK_E$. The space $\mathcal V_T$ is finite-dimensional
and closed. Difference quotients of the trajectories at $t=0$ and their
successive derivatives show that $G^rE_k\in\mathcal V_T$ for all
$r\ge0$. Consequently $\mathcal V_T=\cK_E$.

Equality \eqref{eq:gramian-quadratic-form-ru} gives
$\ker\mathfrak G_T=\mathcal V_T^\perp$. For a self-adjoint operator
$\operatorname{Ran}\mathfrak G_T=(\ker\mathfrak G_T)^\perp$, whence
\eqref{eq:gramian-range-ru} follows.

For an orthogonal projector $P$ we have
\[
 \mathcal J_T(P)
 =\operatorname{Tr}(\mathfrak G_T)
  -\operatorname{Tr}(P\mathfrak G_T).
\]
The Courant--Fischer principle gives the maximum
$\operatorname{Tr}(P\mathfrak G_T)=\sum_{j=1}^q\sigma_j(T)$ at $P_q$.
This proves part (ii). Finally,
\begin{align*}
 &\sum_k\int_0^T
 |\langle\rho_0,(I-P_q)x_k(t)\rangle_{\rm HS}|^2\,dt\\
 &\qquad=
 \langle\rho_0,(I-P_q)\mathfrak G_T(I-P_q)\rho_0\rangle_{\rm HS}\\
 &\qquad\le\sigma_{q+1}(T)\|\rho_0\|_{\rm HS}^2,
\end{align*}
because $P_q$ is a spectral projector of $\mathfrak G_T$.
\end{proof}

The spectrum of $\mathfrak G_T$ defines a stable visibility scale. For example,
the relative effective dimension can be defined as
\begin{equation}
 d_{\rm eff}(T,\varepsilon)
 :=\min\left\{q:\
 \sum_{j>q}\sigma_j(T)
 \le\varepsilon^2\sum_{j\ge1}\sigma_j(T)\right\}.
 \label{eq:effective-visible-dimension-ru}
\end{equation}
For $\varepsilon=0$ it coincides with $\dim\cK_E$, while for positive
$\varepsilon$ it accounts for the weight of the modes on the given time
horizon.

\begin{corollary}[Quadratic weight of new visible modes]
\label{cor:quadratic-gramian-tail-ru}
Let $G_\eta=G_0+\eta V$, $|\eta|\le\eta_0$, and let $\cK_0$ be the exact
visible space for $G_0$. Set $d_0:=\dim\cK_0$ and
\[
 \beta:=\sup_{|\xi|\le\eta_0}\mu_2(G_\xi).
\]
If $\sigma_j^{(\eta)}(T)$ are the eigenvalues of the Gramian for
$G_\eta$, then
\begin{equation}
 \sum_{j>d_0}\sigma_j^{(\eta)}(T)
 \le
 \eta^2\|V\|_{2\to2}^2
 \left(\sum_{k=1}^{m}\|E_k\|_{\rm HS}^2\right)
 \int_0^T t^2e^{2\beta t}\,dt.
 \label{eq:quadratic-gramian-tail-ru}
\end{equation}
For QPE the sum of the norms of the effects on the right-hand side equals $M$.
\end{corollary}

\begin{proof}
Let $P_0$ be the orthogonal projection onto $\cK_0$. By the optimality in
Theorem~\ref{thm:finite-time-gramian-ru}
\begin{equation}
 \sum_{j>d_0}\sigma_j^{(\eta)}(T)
 \le\sum_k\int_0^T
 \|(I-P_0)e^{tG_\eta}E_k\|_{\rm HS}^2\,dt.
 \label{eq:gramian-tail-trial-space-ru}
\end{equation}
Since $e^{tG_0}E_k\in\cK_0$, the Duhamel formula gives
\begin{align*}
 \|(I-P_0)e^{tG_\eta}E_k\|_{\rm HS}
 &\le\|(e^{tG_\eta}-e^{tG_0})E_k\|_{\rm HS}\\
 &\le |\eta|\int_0^t
 \|e^{(t-s)G_\eta}\|_{2\to2}\|V\|_{2\to2}
 \|e^{sG_0}E_k\|_{\rm HS}\,ds\\
 &\le |\eta|t e^{\beta t}\|V\|_{2\to2}\|E_k\|_{\rm HS}.
\end{align*}
Squaring, summing over $k$ and integrating over time in
\eqref{eq:gramian-tail-trial-space-ru} proves
\eqref{eq:quadratic-gramian-tail-ru}.
\end{proof}

For computational reduction one often uses a trial space obtained from a
simplified generator or from a truncated Krylov procedure. The following result
gives an accuracy estimate that can be verified after the construction.

\begin{theorem}[A posteriori stability of approximate reduction]
\label{thm:residual-stability-ru}
Let the space $\cK_0$ contain all effects $E_k$, and let
$P$ be the orthogonal projection onto $\cK_0$ in the
Hilbert--Schmidt metric. Set
\begin{equation}
 A:=PGP|_{\cK_0},
 \qquad
 R:=(I-P)GP|_{\cK_0}.
 \label{eq:residual-operators-ru}
\end{equation}
For the approximate probability
\begin{equation}
 \widetilde p_k(t):=\Tr\!\left[\rho_0e^{tA}(E_k)\right]
 \label{eq:reduced-probability-ru}
\end{equation}
the estimate
\begin{equation}
 |p_k(t)-\widetilde p_k(t)|
 \le
 \norm{\rho_0}_{\rm HS}\norm{E_k}_{\rm HS}
 \norm{R}_{2\to2}\,
 \Psi_t\!\left(\mu_2(G),\mu_2(A)\right),
 \label{eq:residual-bound-ru}
\end{equation}
holds, where
\begin{equation}
 \mu_2(B):=\lambda_{\max}\!\left(\frac{B+B^\dagger}{2}\right),
 \qquad
 \Psi_t(a,b):=
 \begin{cases}
 \dfrac{e^{at}-e^{bt}}{a-b}, & a\ne b,\\[6pt]
 te^{at}, & a=b.
 \end{cases}
 \label{eq:stability-factor-ru}
\end{equation}
For QPE, $\norm{\rho_0}_{\rm HS}\le1$ and
$\norm{E_k}_{\rm HS}=1$.
\end{theorem}

\begin{proof}
On $\cK_0$ the decomposition
\[
 GP=PGP+(I-P)GP=A+R,
\]
holds, where the first term takes values in $\cK_0$ and the second in
$\cK_0^\perp$. Denote $x(t)=e^{tG}E_k$, $y(t)=e^{tA}E_k$ and embed
$y(t)$ in the full operator space. For $z=x-y$ we obtain
\[
 \dot z=Gz+Ry,
 \qquad z(0)=0.
\]
The Duhamel formula gives
\[
 z(t)=\int_0^t e^{(t-s)G}R e^{sA}E_k\,ds.
\]
For completeness we derive the semigroup estimate used. If
$u(t)=e^{tB}u_0$, then
\[
 \frac{d}{dt}\norm{u(t)}_2^2
 =2\operatorname{Re}\langle u(t),Bu(t)\rangle
 \le2\mu_2(B)\norm{u(t)}_2^2.
\]
Gr\"onwall's lemma gives
$\norm{e^{tB}}_{2\to2}\le e^{t\mu_2(B)}$. Therefore
\begin{align*}
 \norm{z(t)}_{\rm HS}
 &\le \norm{R}_{2\to2}\norm{E_k}_{\rm HS}
 \int_0^t e^{(t-s)\mu_2(G)}e^{s\mu_2(A)}\,ds\\
 &=\norm{R}_{2\to2}\norm{E_k}_{\rm HS}
 \Psi_t\bigl(\mu_2(G),\mu_2(A)\bigr).
\end{align*}
Since $\rho_0=\rho_0^\dagger$, we have
\[
 p_k(t)-\widetilde p_k(t)
 =\langle\rho_0,z(t)\rangle_{\rm HS}.
\]
The Cauchy--Schwarz inequality gives \eqref{eq:residual-bound-ru}. For a state
$\Tr(\rho_0^2)\le(\Tr\rho_0)^2=1$, and each QPE effect is a rank-one
projector. Consequently,
\[
 \norm{\rho_0}_{\rm HS}\le1,
 \qquad \norm{E_k}_{\rm HS}=\sqrt{\Tr(E_k^2)}=1.
\]
\end{proof}

\begin{corollary}[Small perturbation of the generator]
\label{cor:small-perturbation-ru}
Let $G_\eta=G_0+\eta V$ with $G_0\cK_0\subseteq\cK_0$. Then
\begin{equation}
 R_\eta=\eta(I-P)VP,
 \label{eq:linear-residual-ru}
\end{equation}
and on every bounded interval of values of $\eta$ the right-hand side of
\eqref{eq:residual-bound-ru} is of order $O(|\eta|)$ for fixed
$t$.
\end{corollary}

\begin{proof}
Invariance gives $(I-P)G_0P=0$. The logarithmic norms of finite-dimensional
operators depend continuously on $\eta$, so the factor
\eqref{eq:stability-factor-ru} is bounded on a compact interval.
\end{proof}

\begin{proposition}[Quadratic estimate for a reducing subspace]
\label{prop:quadratic-residual-ru}
Let $P$ be the orthogonal projection onto $\cK_0$,
$G_\eta=G_0+\eta V$, $G_0P=PG_0$, and let
$P\rho_0=\rho_0$ and $PE_k=E_k$. For a fixed $\eta_0>0$ set
\[
 A_\xi:=P(G_0+\xi V)P\big|_{\cK_0},
 \quad
 \beta_G:=\sup_{|\xi|\le\eta_0}\mu_2(G_0+\xi V),
 \quad
 \beta_A:=\sup_{|\xi|\le\eta_0}\mu_2(A_\xi).
\]
Then for $|\eta|\le\eta_0$ and $t\ge0$
\begin{align}
 &\left|\Tr\!\left[\rho_0e^{tG_\eta}(E_k)\right]
 -\Tr\!\left[\rho_0e^{tA_\eta}(E_k)\right]\right|
 \nonumber\\
 &\qquad\le
 \frac{\eta^2t^2}{2}\norm{\rho_0}_{\rm HS}\norm{E_k}_{\rm HS}
 \left(
 \norm{V}_{2\to2}^2e^{t\beta_G}
 +\norm{PVP}_{2\to2}^2e^{t\beta_A}
 \right).
 \label{eq:quadratic-residual-bound-ru}
\end{align}
\end{proposition}

\begin{proof}
Set
\[
 f(\eta):=
 \left\langle\rho_0,
 \bigl(e^{tG_\eta}-e^{tA_\eta}\bigr)E_k
 \right\rangle_{\rm HS}.
\]
From $G_0P=PG_0$ follow the invariance of both orthogonal summands,
$A_0=G_0|_{\cK_0}$ and the equality
\[
 e^{sA_0}P=e^{sG_0}P.
\]
In particular, $f(0)=0$. For a family $B_\xi=B_0+\xi W$ the Duhamel formula
gives
\begin{equation}
 \frac{d}{d\xi}e^{tB_\xi}
 =\int_0^t e^{(t-s)B_\xi}W e^{sB_\xi}\,ds.
 \label{eq:first-duhamel-ru}
\end{equation}
Applying it to $G_\eta$ and $A_\eta$ at $\eta=0$, we obtain
\[
 \left.\frac{d}{d\eta}\right|_{\eta=0}
 \bigl(e^{tG_\eta}-e^{tA_\eta}\bigr)E_k
 =\int_0^t e^{(t-s)G_0}(I-P)VP e^{sG_0}E_k\,ds.
\]
Each integrand vector belongs to $\cK_0^\perp$, because
$e^{(t-s)G_0}$ preserves this space. The condition
$\rho_0\in\cK_0$ gives $f'(0)=0$.

Repeated application of \eqref{eq:first-duhamel-ru} gives the exact
time-ordered formula
\begin{equation}
 \frac{d^2}{d\xi^2}e^{tB_\xi}
 =2\int_0^t\int_0^s
 e^{(t-s)B_\xi}W e^{(s-u)B_\xi}W e^{uB_\xi}\,du\,ds.
 \label{eq:second-duhamel-ru}
\end{equation}
The semigroup estimate from the proof of
Theorem~\ref{thm:residual-stability-ru} and the area of the simplex
$0\le u\le s\le t$ lead to
\begin{align*}
 \sup_{|\xi|\le\eta_0}
 \left\|\frac{d^2}{d\xi^2}e^{tG_\xi}\right\|_{2\to2}
 &\le t^2\norm{V}_{2\to2}^2e^{t\beta_G},\\
 \sup_{|\xi|\le\eta_0}
 \left\|\frac{d^2}{d\xi^2}e^{tA_\xi}\right\|_{2\to2}
 &\le t^2\norm{PVP}_{2\to2}^2e^{t\beta_A}.
\end{align*}
Consequently, for $|\xi|\le\eta_0$
\begin{align*}
 |f''(\xi)|
 \le t^2\norm{\rho_0}_{\rm HS}\norm{E_k}_{\rm HS}
 \left(
 \norm{V}_{2\to2}^2e^{t\beta_G}
 +\norm{PVP}_{2\to2}^2e^{t\beta_A}
 \right).
\end{align*}
The integral Taylor remainder has the form
\[
 f(\eta)=\int_0^\eta(\eta-\xi)f''(\xi)\,d\xi
\]
for $\eta\ge0$; for $\eta<0$ the same integral with reversed orientation is
used. In both cases the integral of the modulus of the weight function equals
$\eta^2/2$. Substituting the obtained bound for $f''$ proves
\eqref{eq:quadratic-residual-bound-ru}.
\end{proof}

The quadratic order has a geometric meaning. When $G_0P=PG_0$ the first-order
leakage lies in $\cK_0^\perp$ and has zero pairing with the compatible state
$\rho_0\in\cK_0$. To influence a probability, this component must return to
$\cK_0$ through a repeated action of the perturbation. The conditions on the
state and on the measurement ensure exactly this first-order cancellation.

\begin{corollary}[Visible spectral regimes]
An oscillating component of a probability is possible only if $A_E$ has a
visible eigenvalue $a+ib$ with $b\ne0$. The oscillation decays for
$a<0$ and keeps its amplitude for $a=0$. Invisibility is characterized by
orthogonality to the observable space: a state difference
$\Delta\rho\in\cK_E^\perp$ gives zero change in every outcome probability
at every time.
\end{corollary}

\begin{proof}
By Theorem~\ref{thm:visible-reduction-ru} each probability is a sum of Jordan
contributions of the visible restriction. A complex-conjugate pair gives the
real terms $t^re^{at}\cos bt$ and $t^re^{at}\sin bt$.
The adjoint quantum Markov semigroup is completely positive and unital,
whence
$\norm{e^{tG}}_{\infty\to\infty}=1$ for $t\ge0$.
Boundedness of the semigroup implies $\operatorname{Re}\lambda\le0$.
A Jordan block of size greater than one at $\operatorname{Re}\lambda=0$ would
produce polynomial growth $t^r$, contradicting this boundedness. Hence the
peripheral blocks are semisimple. Finally, $e^{tG}E_k\in\cK_E$ implies
$\Tr[\Delta\rho\,e^{tG}E_k]=0$ for every Hermitian
$\Delta\rho\in\cK_E^\perp$, which proves the invisibility statement.
\end{proof}

\begin{proposition}[Operational equivalence]
\label{prop:equivalence-ru}
Let $\cL_1^\dagger$ and $\cL_2^\dagger$ preserve a common subspace
$\cK$ containing all $E_k$ and coincide on $\cK$. Then they give the same
outcome probabilities for any initial state and any $t\ge0$.
\end{proposition}

\begin{proof}
Denote the common restriction by
$A=\cL_1^\dagger|_{\cK}=\cL_2^\dagger|_{\cK}$. Invariance gives
$(\cL_1^\dagger)^rE_k=A^rE_k=(\cL_2^\dagger)^rE_k$ for each $r\ge0$.
Comparing the power series of the exponentials yields
$e^{t\cL_1^\dagger}E_k=e^{tA}E_k=e^{t\cL_2^\dagger}E_k$. Substituting into
\eqref{eq:heisenberg-probability-ru} completes the proof.
\end{proof}

This proposition fixes the class of operational equivalence: QPE identifies
generators that coincide on the visible subspace.

\begin{remark}[Closure on the probability vector]
If the $E_k$ are mutually orthogonal projectors and $\sum_kE_k=I$, then the
equality $\dot p=Qp$ with a single matrix $Q$ for all initial states is
equivalent to invariance of $\Span\{E_k\}$ under $G$; in that case $Q$
is the generator of a Markov chain. For general QPE noise this hull may be
non-invariant, and $\cK_{\rm QPE}$ adds the necessary intermediate coherences.
\end{remark}

\section{Exact QPE-visible projection}

Represent an operator of the counting register in the form
$X=\sum_{j,\ell=0}^{M-1}X_{j\ell}\ket j\bra\ell$. For
$r\in\mathbb Z_M$ introduce the cyclic sum of a diagonal
\begin{equation}
 c_r(X):=
 \sum_{\substack{0\le j,\ell<M\\j-\ell\equiv r\;({\rm mod}\,M)}}X_{j\ell}.
 \label{eq:cyclic-diagonal-ru}
\end{equation}

\begin{theorem}[Visibility of cyclic diagonals]
\label{thm:cyclic-visibility-ru}
The QPE probability vector and the cyclic sums of the diagonals are related by
the discrete Fourier transform:
\begin{align}
 p_k&=\frac1M\sum_{r=0}^{M-1}c_r(\rho_C)e^{-2\pi ikr/M},
 \label{eq:p-from-cyclic-ru}\\
 c_r(\rho_C)&=\sum_{k=0}^{M-1}p_ke^{2\pi ikr/M}.
 \label{eq:cyclic-from-p-ru}
\end{align}
Consequently, a perturbation $X$ is instantaneously invisible to QPE if and
only if
\begin{equation}
 c_r(X)=0\qquad\text{for all }r\in\mathbb Z_M.
 \label{eq:instant-invisible-ru}
\end{equation}
\end{theorem}

\begin{proof}
The matrix elements of the effect are
$(E_k)_{\ell j}=M^{-1}e^{-2\pi ik(j-\ell)/M}$. Therefore
\begin{align*}
 p_k
 &=\sum_{j,\ell}(\rho_C)_{j\ell}(E_k)_{\ell j}\\
 &=\frac1M\sum_{r=0}^{M-1}
 \left(\sum_{j-\ell\equiv r\ ({\rm mod}\,M)}
 (\rho_C)_{j\ell}\right)e^{-2\pi ikr/M},
\end{align*}
which coincides with \eqref{eq:p-from-cyclic-ru}. Orthogonality of the
characters
\[
 \frac1M\sum_{k=0}^{M-1}e^{2\pi ik(r-s)/M}=\mathbf1_{r=s}
\]
gives the inverse formula \eqref{eq:cyclic-from-p-ru}. The discrete Fourier
transform is invertible, so $p_k(X)=0$ for all $k$ is equivalent to the
conditions \eqref{eq:instant-invisible-ru}.
\end{proof}

Dynamical invisibility is characterized by orthogonality to the entire space
$\cK_{\rm QPE}$.

\subsection{Representation via coherence transfer}

Coherences with respect to the computational basis are the off-diagonal
elements of the state matrix; for a pure state
$\sum_j a_j\ket j$ they equal $a_j\overline{a_\ell}$ and contain the relative
phases of the amplitudes.
For an input eigenstate, assume that noise acts on the counting register as
a population-preserving entrywise-multiplication (Schur) channel. This gives
the reduced matrix
\begin{equation}
 (\rho_C)_{j\ell}
 =\frac1M e^{2\pi i\omega(j-\ell)}C_{j\ell},
 \qquad C_{jj}=1.
 \label{eq:coherence-transfer-ru}
\end{equation}
Here $C=(C_{j\ell})$ is the coherence-transfer matrix.
The condition $C_{jj}=1$ expresses preservation of the populations.
Complete positivity of the entrywise-multiplication channel is equivalent to
$C\succeq0$. Sufficiency follows from the Schur product theorem for positive
matrices. In the opposite direction, the Choi matrix of such a channel is
supported on
$\Span\{\ket j\otimes\ket j\}_{j=0}^{M-1}$, and its matrix in the orthonormal
basis $\{\ket j\otimes\ket j\}$ equals $C$. Therefore positivity of the Choi
matrix implies $C\succeq0$.
The exact QPE law has the form
\begin{equation}
 p_k=\frac{1}{M^2}\sum_{j,\ell=0}^{M-1}
 C_{j\ell}e^{2\pi i(j-\ell)(\omega-k/M)}.
 \label{eq:general-transfer-law-ru}
\end{equation}
This formula expresses the outcome probabilities through the individual
coherence-transfer coefficients.

If $C_{j\ell}=g_{j-\ell}$ is a Toeplitz matrix and
$g_{-d}=\overline{g_d}$, then
\begin{equation}
 p_k=\frac{1}{M^2}\left[
 M+2\sum_{d=1}^{M-1}(M-d)
 \operatorname{Re}\!\left(g_de^{2\pi id(\omega-k/M)}\right)
 \right].
 \label{eq:toeplitz-law-ru}
\end{equation}

\begin{proposition}[Equal Ramsey contrast and different QPE outcomes]
\label{prop:ramsey-separation-ru}
For $M=4$ there exist two random-unitary phase channels with the common value
$g_1=1/2$ which at $\omega=0$ give the distributions
\begin{equation}
 (p_0^{(+)},p_1^{(+)},p_2^{(+)},p_3^{(+)})
 =\left(\frac34,0,\frac14,0\right),
 \qquad
 (p_0^{(-)},p_1^{(-)},p_2^{(-)},p_3^{(-)})
 =\left(\frac12,\frac14,0,\frac14\right).
 \label{eq:ramsey-separated-qpe-ru}
\end{equation}
Their single-step Ramsey contrasts coincide, whereas the supports of the QPE
distributions and the outcome weights differ.
\end{proposition}

\begin{proof}
Let $N=\sum_{j=0}^{3}j\proj j$ and
\begin{equation}
 \Phi_\nu(\rho)=\int_0^1
 e^{2\pi i\theta N}\rho e^{-2\pi i\theta N}\,d\nu(\theta).
 \label{eq:random-phase-channel-ru}
\end{equation}
This is a random-unitary channel with
$g_d=\int_0^1e^{2\pi id\theta}\,d\nu(\theta)$. Take
\[
 \nu_+=\tfrac34\delta_0+\tfrac14\delta_{1/2},
 \qquad
 \nu_-=\tfrac12\delta_0+\tfrac14\delta_{1/4}
 +\tfrac14\delta_{3/4}.
\]
Then
$g_d^{(+)}=3/4+(-1)^d/4$ and
$g_d^{(-)}=1/2+\cos(\pi d/2)/2$. For $d=1,2,3$ one obtains the sets
$(1/2,1,1/2)$ and $(1/2,0,1/2)$. For real $g_d$, $M=4$ and
$\omega=0$ formula \eqref{eq:toeplitz-law-ru} takes the form
\[
 p_k=\frac1{16}\left[
 4+6g_1\cos\frac{\pi k}{2}
 +4g_2\cos(\pi k)
 +2g_3\cos\frac{3\pi k}{2}
 \right].
\]
For the two sets of harmonics this gives, respectively,
\[
 16p^{(+)}=(12,0,4,0),
 \qquad
 16p^{(-)}=(8,4,0,4),
\]
which proves \eqref{eq:ramsey-separated-qpe-ru}.
\end{proof}

Proposition~\ref{prop:ramsey-separation-ru} demonstrates that the QPE
distribution depends on cyclic linear combinations of several harmonics $g_d$.

For an exact description of this dependence set
\begin{equation}
 \widetilde g_d:=g_de^{2\pi id\omega},
 \qquad h_0:=M,
 \qquad
 h_r:=(M-r)\widetilde g_r+r\widetilde g_{r-M},
 \quad r=1,\ldots,M-1.
 \label{eq:observable-harmonic-combinations-ru}
\end{equation}
For $1\le r\le M-1$ among the pairs with $j-\ell\equiv r\pmod M$ there are
$M-r$ pairs with $j-\ell=r$ and $r$ pairs with $j-\ell=r-M$. Therefore
$c_0(\rho_C)=1$, $c_r(\rho_C)=h_r/M$, and
Theorem~\ref{thm:cyclic-visibility-ru} gives
\begin{equation}
 p_k=\frac1{M^2}\sum_{r=0}^{M-1}h_re^{-2\pi ikr/M}.
 \label{eq:qpe-harmonic-combinations-ru}
\end{equation}

Invertibility of the discrete Fourier transform in
\eqref{eq:qpe-harmonic-combinations-ru} shows that two Toeplitz sequences give
the same QPE distribution at fixed $M$ and $\omega$ exactly when all $h_r$
coincide. For $g_d=a_d\ge0$ the kernel is even and maximal at $k/M=\omega$,
and the factor $e^{2\pi id\delta}$ shifts its centre to $\omega+\delta$.
Changes that preserve the cyclic sums are operationally invisible.

Thus the algorithmic phase error $\epsilon_{\rm alg}$ is represented by a
coherent shift $\delta$, Markov decoherence changes the moduli $a_d$, and the
readout error acts already on the finished vector $p$.

\section{Damped Fej\'er kernel and the visible coherence length}

Consider the Markov law of contrast transfer
\begin{equation}
 g_d=e^{-\alpha|d|}e^{2\pi id\delta},
 \qquad \alpha\ge0.
 \label{eq:markov-transfer-ru}
\end{equation}
Here $\alpha$ is the dimensionless integrated damping per unit distance
between time branches, and $\delta$ is the coherent phase detuning.
Substitution into \eqref{eq:toeplitz-law-ru} gives
\begin{equation}
 p_k^{(\alpha,\delta)}
 =\frac{1}{M^2}\left[
 M+2\sum_{d=1}^{M-1}(M-d)e^{-\alpha d}
 \cos 2\pi d\left(\omega+\delta-\frac{k}{M}\right)
 \right].
 \label{eq:damped-fejer-ru}
\end{equation}

First we compute the visible space for a general generator that is diagonal on
the matrix units. Such quantum Markov semigroups are semigroups of Schur
multipliers; their completely positive realizations through conditionally
negative definite kernels belong to the general theory
\cite{WirthZhang2023}. The following theorem covers inhomogeneous purely
dephasing semigroups, in which the damping rate depends on both indices of the
time branches, and computes the QPE-visible part by spectral classes.

\begin{theorem}[Spectral formula for the QPE-visible dimension]
\label{thm:diagonal-generator-dimension-ru}
Let $e_{j\ell}:=\ket j\bra\ell$ and let the adjoint generator $G$ be
diagonal in the basis of matrix units:
\begin{equation}
 G(e_{j\ell})=\lambda_{j\ell}e_{j\ell},
 \qquad 0\le j,\ell<M.
 \label{eq:diagonal-generator-ru}
\end{equation}
For $r\in\mathbb Z_M$ set
\begin{align}
 \Lambda_r
 &:=\{\lambda_{j\ell}:j-\ell\equiv r\pmod M\},
 \label{eq:residue-spectrum-ru}\\
 B_{r,\zeta}
 &:=\sum_{\substack{0\le j,\ell<M\\
 j-\ell\equiv r\ ({\rm mod}\,M)\\
 \lambda_{j\ell}=\zeta}}e_{j\ell},
 \qquad \zeta\in\Lambda_r.
 \label{eq:residue-eigenmode-ru}
\end{align}
Then the minimal QPE-visible space has the exact decomposition
\begin{equation}
 \cK_{\rm QPE}
 =\bigoplus_{r\in\mathbb Z_M}
 \Span\{B_{r,\zeta}:\zeta\in\Lambda_r\},
 \qquad
 \dim_{\mathbb C}\cK_{\rm QPE}
 =\sum_{r\in\mathbb Z_M}|\Lambda_r|.
 \label{eq:general-diagonal-visible-space-ru}
\end{equation}
If $G$ is the adjoint generator of a quantum Markov semigroup,
then $\lambda_{jj}=0$, hence $|\Lambda_0|=1$ and
\begin{equation}
 M\le\dim_{\mathbb C}\cK_{\rm QPE}\le M^2-M+1.
 \label{eq:diagonal-visible-bounds-ru}
\end{equation}
Both bounds are attained in the class of diagonal generators.
\end{theorem}

\begin{proof}
For each residue introduce the orthogonal coordinate space
\begin{equation}
 \mathcal D_r
 :=\Span\{e_{j\ell}:j-\ell\equiv r\pmod M\}.
\end{equation}
These spaces form a direct decomposition of $\mathcal B(\cH_C)$ and are
invariant under $G$. The discrete Fourier transform of the effects
\eqref{eq:qpe-effects-ru} gives
\begin{equation}
 F_r:=\sum_{k=0}^{M-1}e^{-2\pi ikr/M}E_k
 =\sum_{j-\ell\equiv r\ ({\rm mod}\,M)}e_{j\ell}
 \in\mathcal D_r.
 \label{eq:general-residue-effect-ru}
\end{equation}
Since this transform is invertible,
\begin{equation}
 \cK_{\rm QPE}
 =\bigoplus_{r\in\mathbb Z_M}
 \Span\{G^sF_r:s\ge0\}.
 \label{eq:residue-krylov-sum-ru}
\end{equation}
Diagonality of the generator implies
\begin{equation}
 G^sF_r
 =\sum_{\zeta\in\Lambda_r}\zeta^sB_{r,\zeta}.
 \label{eq:residue-vandermonde-ru}
\end{equation}
Let $q_r=|\Lambda_r|$. For $s=0,\ldots,q_r-1$ the coefficients in
\eqref{eq:residue-vandermonde-ru} form a Vandermonde matrix on the distinct
numbers $\zeta\in\Lambda_r$. Its determinant is nonzero. Therefore the
Krylov span of $F_r$ coincides with the linear span of the $q_r$ nonzero
matrices $B_{r,\zeta}$. Their supports are pairwise disjoint, and the spaces
$\mathcal D_r$ form a direct sum. This proves
\eqref{eq:general-diagonal-visible-space-ru}.

Each mode can be isolated constructively by means of the Lagrange
interpolation polynomial:
\begin{equation}
 B_{r,\zeta}
 =\left[\prod_{\substack{\xi\in\Lambda_r\\\xi\ne\zeta}}
 \frac{G-\xi\id}{\zeta-\xi}\right]F_r.
 \label{eq:lagrange-visible-mode-ru}
\end{equation}
Here $\id$ is the identity operator on the space of matrices, and the empty
product equals $\id$. The polynomial takes the value one at $\zeta$ and zero
at the other elements of $\Lambda_r$; its degree $q_r-1$ explicitly bounds the
number of powers of $G$ required.

Unitality of the adjoint Markov semigroup gives $G(I)=0$. By
\eqref{eq:diagonal-generator-ru} and the linear independence of the $e_{jj}$ we
obtain $\lambda_{jj}=0$. Each residue class contains $M$ matrix units, whence
$1\le|\Lambda_r|\le M$ for $r\ne0$, and the bounds
\eqref{eq:diagonal-visible-bounds-ru} follow. The lower bound is attained at
$G=0$. The upper bound is attained for a purely Hamiltonian generator with
$H=\sum_{j=0}^{M-1}2^j\proj j$: the numbers
$i(2^j-2^\ell)$ are pairwise distinct for all ordered pairs $j\ne\ell$.
Indeed, the sign of the difference fixes the order of the indices, its
$2$-adic valuation equals the smaller index, and the remaining odd factor
determines their difference.
\end{proof}

\begin{example}[Coinciding frequencies in different cyclic classes]
\label{ex:eight-visible-modes-ru}
Let $M=4$, $H=\operatorname{diag}(0,1,2,1)$ and $G(X)=i[H,X]$
in units $\hbar=1$. Then $\lambda_{j\ell}=i(h_j-h_\ell)$, whence
\[
\begin{array}{c|c|c}
r&\Lambda_r&|\Lambda_r|\\\hline
0&\{0\}&1\\
1&\{i,-i\}&2\\
2&\{2i,0,-2i\}&3\\
3&\{i,-i\}&2
\end{array}
\qquad \dim\cK_{\rm QPE}=1+2+3+2=8.
\]
For example, in the class $r=1$ set
$B_+=e_{10}+e_{21}$, $B_-=e_{32}+e_{03}$.
We have $F_1=B_++B_-$, $GF_1=iB_+-iB_-$ and hence
\[
 B_\pm=\tfrac12(F_1\mp iGF_1),\qquad
 e^{tG}F_1=e^{it}B_++e^{-it}B_-.
\]
In the class $r=3$ the same frequencies $\pm i$ correspond to different matrix
supports. In total the generator has five distinct eigenvalues, while the
visible space has eight independent directions. Thus frequencies must be
counted separately within each cyclic class. This explains the sum of
cardinalities in \eqref{eq:general-diagonal-visible-space-ru}.
\end{example}

\begin{example}[Independent qubit dephasing]
\label{ex:local-dephasing-ru}
Let $M=2^n$, let $j_a\in\{0,1\}$ be a digit of the number $j$,
$Z_a\ket j=(-1)^{j_a}\ket j$, $L_a=\sqrt{\gamma_a/2}\,Z_a$,
$\gamma_a\ge0$, and $H=\sum_jh_j\proj j$. Then
\[
 L_a^\dagger e_{j\ell}L_a
 -\tfrac12\{L_a^\dagger L_a,e_{j\ell}\}
 =-\gamma_a\mathbf1_{j_a\ne\ell_a}e_{j\ell},
\]
and therefore
\[
 G(e_{j\ell})=\lambda_{j\ell}e_{j\ell},\qquad
 \lambda_{j\ell}=i(h_j-h_\ell)
 -\sum_{a=0}^{n-1}\gamma_a\mathbf1_{j_a\ne\ell_a}.
\]
The spectral theorem applies for any rates $\gamma_a$.
For $n=2$, $H=0$ and $\gamma_0=\gamma_1=\gamma>0$ we have
\[
 \Lambda_0=\{0\},\quad
 \Lambda_1=\Lambda_3=\{-\gamma,-2\gamma\},\quad
 \Lambda_2=\{-\gamma\},\qquad \dim\cK_{\rm QPE}=6.
\]
Here the pairs $(1,0)$ and $(2,1)$ have index difference one and Hamming
distances one and two, respectively. The distinct damping rates on this
ordinary diagonal enter the exact dimension formula through the spectral
classes $\Lambda_r$.
\end{example}

\begin{corollary}[Spectral blocks of the Gramian]
\label{cor:diagonal-gramian-blocks-ru}
Under the conditions of Theorem~\ref{thm:diagonal-generator-dimension-ru} the
Gramian $\mathfrak G_T$ decomposes along the spaces $\mathcal D_r$. In the
orthonormal basis
$\{e_{j\ell}:j-\ell\equiv r\pmod M\}$ its $r$-th block has the matrix
\begin{equation}
 \left[\mathfrak G_{T,r}\right]_{(j,\ell),(p,q)}
 =\frac1M\,\varphi_T\!\left(
 \lambda_{j\ell}+\overline{\lambda_{pq}}\right),
 \label{eq:diagonal-gramian-block-ru}
\end{equation}
where both pairs of indices belong to the class $r$ and
\begin{equation}
 \varphi_T(z):=
 \begin{cases}
  (e^{Tz}-1)/z,&z\ne0,\\
  T,&z=0.
 \end{cases}
 \label{eq:gramian-scalar-kernel-ru}
\end{equation}
Moreover,
\begin{equation}
 \operatorname{rank}\mathfrak G_{T,r}=|\Lambda_r|.
 \label{eq:diagonal-gramian-rank-ru}
\end{equation}
\end{corollary}

\begin{proof}
The normalized discrete Fourier transform of the outcomes replaces the effects
$E_k$ by orthogonal linear combinations
\[
 \widehat F_r:=\frac1{\sqrt M}F_r
 =\frac1{\sqrt M}
 \sum_{j-\ell\equiv r\ ({\rm mod}\,M)}e_{j\ell}.
\]
A unitary change of the generating family preserves the sum of rank-one
operators in \eqref{eq:finite-time-gramian-ru}. In addition,
\[
 e^{tG}\widehat F_r
 =\frac1{\sqrt M}
 \sum_{j-\ell\equiv r\ ({\rm mod}\,M)}
 e^{t\lambda_{j\ell}}e_{j\ell}\in\mathcal D_r.
\]
The spaces $\mathcal D_r$ are orthogonal, so the Gramian is a direct sum of
blocks. Integrating the products of coordinates gives
\eqref{eq:diagonal-gramian-block-ru}. Equality of ranks follows from
\eqref{eq:gramian-range-ru} and the decomposition
\eqref{eq:general-diagonal-visible-space-ru}.
\end{proof}

\begin{corollary}[Generator depending on the ordinary diagonal]
\label{cor:ordinary-diagonal-dimension-ru}
If $G(e_{j\ell})=\mu_{j-\ell}e_{j\ell}$ and $\mu_0=0$, then
\begin{equation}
 \dim_{\mathbb C}\cK_{\rm QPE}
 =1+\sum_{r=1}^{M-1}\#\{\mu_r,\mu_{r-M}\}
 =M+\#\{r:\mu_r\ne\mu_{r-M}\}.
 \label{eq:ordinary-diagonal-visible-dimension-ru}
\end{equation}
\end{corollary}

\begin{proof}
In the zero class only the eigenvalue $\mu_0$ occurs. For
$1\le r<M$ the ordinary index difference in the class $r$ takes two values:
$r$ and $r-M$. It remains to apply
Theorem~\ref{thm:diagonal-generator-dimension-ru}.
\end{proof}

The meaning of this count is easiest to trace on a Toeplitz matrix
$X_{j\ell}=x_{j-\ell}$. For $1\le r<M$ the cyclic sum merges
two ordinary diagonals of lengths $M-r$ and $r$:
\[
 c_r(X)=(M-r)x_r+r x_{r-M},\qquad
 c_r(e^{tG}X)=(M-r)e^{t\mu_r}x_r+r e^{t\mu_{r-M}}x_{r-M}.
\]
If $\mu_r\ne\mu_{r-M}$, the two components have distinct exponential
time dependences. For example, adding
$(ra,-(M-r)a)$ to $(x_r,x_{r-M})$ preserves $c_r(X)$
and changes $c_r(e^{tG}X)$ by
$r(M-r)a(e^{t\mu_r}-e^{t\mu_{r-M}})$.
The separation of these two components contributes an additional direction
to the Krylov closure.
The effects $E_k$ are Toeplitz, and a generator depending only on $j-\ell$
preserves this property. Therefore the visible space is contained
in the Toeplitz space of dimension $2M-1$;
the formulas hold for every initial state matrix.

\begin{proposition}[GKLS realization and dimension under degeneracies]
\label{prop:gkls-toeplitz-ru}
Let
\begin{equation}
 N=\sum_{j=0}^{M-1}j\proj j,
 \qquad H_\delta=-2\pi\delta N,
 \qquad L_a=\sqrt{2\alpha}\sum_{j=a}^{M-1}\proj j,
 \quad a=1,\ldots,M-1,
 \label{eq:brownian-jumps-ru}
\end{equation}
and
\begin{equation}
 \cL(\rho)=-i[H_\delta,\rho]
 +\sum_{a=1}^{M-1}\left(
 L_a\rho L_a^\dagger-\frac12\{L_a^\dagger L_a,\rho\}
 \right).
 \label{eq:brownian-gkls-ru}
\end{equation}
Then
\begin{equation}
 \bigl(e^{t\cL}\rho\bigr)_{j\ell}
 =e^{-\alpha t|j-\ell|}
  e^{2\pi i\delta t(j-\ell)}\rho_{j\ell}.
 \label{eq:gkls-entrywise-ru}
\end{equation}
Set $\mu_d=-\alpha|d|-2\pi i\delta d$ and
\begin{equation}
 q_r:=\#\{\mu_r,\mu_{r-M}\},
 \qquad r=1,\ldots,M-1.
\end{equation}
Then
\begin{equation}
 \dim_{\mathbb C}\cK_{\rm QPE}
 =1+\sum_{r=1}^{M-1}q_r
 =M+\#\{r:\mu_r\ne\mu_{r-M}\}.
 \label{eq:visible-degenerate-dimension-ru}
\end{equation}
In particular, for $M=2^n$ we obtain
\begin{equation}
 \dim_{\mathbb C}\cK_{\rm QPE}=
 \begin{cases}
 M, & \alpha=\delta=0,\\
 2M-2, & \alpha>0,\ \delta=0,\\
 2M-1, & \delta\ne0.
 \end{cases}
 \label{eq:visible-dimension-cases-ru}
\end{equation}
If $\mu_r\ne\mu_{r-M}$ for all $r=1,\ldots,M-1$, then
\begin{equation}
 \cK_{\rm QPE}=\mathcal T_M,
 \qquad \dim_{\mathbb C}\cK_{\rm QPE}=2M-1,
 \label{eq:visible-toeplitz-dimension-ru}
\end{equation}
where $\mathcal T_M$ is the space of complex Toeplitz matrices.
\end{proposition}

\begin{proof}
Denote
$\lambda_a(j)=\sqrt{2\alpha}\,\mathbf1_{j\ge a}$, so that
$L_a\ket j=\lambda_a(j)\ket j$. The coefficient of the dissipator at
$\rho_{j\ell}$ equals
\begin{align*}
 \sum_{a=1}^{M-1}
 \left(\lambda_a(j)\lambda_a(\ell)
 -\frac{\lambda_a(j)^2+\lambda_a(\ell)^2}{2}\right)
 &=-\frac12\sum_{a=1}^{M-1}
 \bigl(\lambda_a(j)-\lambda_a(\ell)\bigr)^2\\
 &=-\alpha|j-\ell|.
\end{align*}
The last equality follows from the fact that exactly $|j-\ell|$ thresholds $a$
separate the indices $j$ and $\ell$. The Hamiltonian part gives
\[
 -i[H_\delta,\rho]_{j\ell}
 =2\pi i\delta(j-\ell)\rho_{j\ell}.
\]
Consequently each matrix element satisfies an independent scalar equation
whose solution coincides with \eqref{eq:gkls-entrywise-ru}.

For the adjoint generator the matrix units are eigenvectors:
\begin{equation}
 G(e_{j\ell})
 =\mu_{j-\ell}e_{j\ell},
 \qquad
 \mu_d=-\alpha|d|-2\pi i\delta d.
\end{equation}
Therefore formula \eqref{eq:visible-degenerate-dimension-ru} coincides with
\eqref{eq:ordinary-diagonal-visible-dimension-ru} of
Corollary~\ref{cor:ordinary-diagonal-dimension-ru}.

The equality $\mu_r=\mu_{r-M}$ splits into the conditions
\[
 \alpha r=\alpha(M-r),
 \qquad \delta r=\delta(r-M).
\]
The second condition gives $\delta=0$. For $\alpha=0$ all pairs coincide,
while for $\alpha>0$ the unique pair $r=M/2$ coincides. For $M=2^n$ this gives
the dimensions $M$, $2M-2$ and $2M-1$ of
\eqref{eq:visible-dimension-cases-ru}, respectively. When all pairs are
separated, the spectral sums $B_{r,\mu_r}$ and $B_{r,\mu_{r-M}}$ of
\eqref{eq:residue-eigenmode-ru} are the ordinary diagonals
$T_r$ and $T_{r-M}$. Together with $T_0$ they form a basis of
$\mathcal T_M$, which proves \eqref{eq:visible-toeplitz-dimension-ru}.
\end{proof}

The threshold operators $L_a$ in \eqref{eq:brownian-jumps-ru} define an
analytically solvable reference GKLS model directly on the counting register.

Formula \eqref{eq:damped-fejer-ru} is an exact consequence of
\eqref{eq:brownian-gkls-ru}.

\begin{proposition}[Crossover of the main peak]
\label{thm:peak-crossover-ru}
Let $\omega+\delta=k_*/M$ and $\alpha>0$. Then
\begin{equation}
 p_{k_*}=\frac{1}{M^2}\left[
 M+2\sum_{d=1}^{M-1}(M-d)e^{-\alpha d}
 \right].
 \label{eq:peak-exact-ru}
\end{equation}
For $\alpha M\ll1$
\begin{equation}
 p_{k_*}=1-\frac{\alpha(M^2-1)}{3M}+O(\alpha^2M^2),
 \label{eq:peak-short-ru}
\end{equation}
and for fixed $\alpha$ and $M\to\infty$
\begin{equation}
 p_{k_*}=\frac{\coth(\alpha/2)}{M}+O(M^{-2}).
 \label{eq:peak-long-ru}
\end{equation}
\end{proposition}

\begin{proof}
Formula \eqref{eq:peak-exact-ru} follows from
\eqref{eq:damped-fejer-ru}. For $x\ge0$
\[
 0\le e^{-x}-1+x\le\frac{x^2}{2}.
\]
Using the identities
\[
 \sum_{d=1}^{M-1}d(M-d)=\frac{M(M^2-1)}6,
 \qquad
 \sum_{d=1}^{M-1}d^2(M-d)=\frac{M^2(M^2-1)}{12},
\]
we obtain
\[
 p_{k_*}=1-\frac{\alpha(M^2-1)}{3M}+R_{\alpha,M},
 \qquad
 0\le R_{\alpha,M}\le\frac{\alpha^2(M^2-1)}{12}.
\]
This proves \eqref{eq:peak-short-ru} uniformly in the regime of bounded
$\alpha M$. For the second limit set $q=e^{-\alpha}\in(0,1)$. The exact
formula for the finite geometric sum gives
\[
 \sum_{d=1}^{M-1}(M-d)q^d
 =M\frac{q}{1-q}-\frac{q}{(1-q)^2}
 +\frac{q^{M+1}}{(1-q)^2}.
\]
After substitution we obtain
\[
 p_{k_*}=\frac1M\left(1+\frac{2q}{1-q}\right)
 -\frac{2q}{M^2(1-q)^2}
 +\frac{2q^{M+1}}{M^2(1-q)^2}.
\]
The equality $1+2q/(1-q)=\coth(\alpha/2)$ proves
\eqref{eq:peak-long-ru}; the constant in the remainder $O(M^{-2})$ depends only
on the fixed $\alpha$.
\end{proof}

The natural crossover scale is
\begin{equation}
 M_{\rm coh}\asymp\alpha^{-1},
 \qquad n_{\rm coh}\asymp\log_2(\alpha^{-1}).
 \label{eq:coherence-scale-ru}
\end{equation}
This scale describes the crossover of the main peak in the weak-damping
regime. For fixed $\alpha>0$ and increasing $M$, the phase-grid spacing is
$M^{-1}$ and the peak probability has the asymptotics
\eqref{eq:peak-long-ru}.

\subsection{Preparation and classical readout}

Let $B$ be a stochastic matrix of classical readout. Then the observed
distribution equals
\begin{equation}
 p^{\rm obs}=Bp^{(\alpha,\delta)}.
 \label{eq:readout-map-ru}
\end{equation}
For an independent symmetric flip of each of the $n$ bits with probability
$e_{\rm ro}$
\begin{equation}
 B=B_{e_{\rm ro}}^{\otimes n},
 \qquad
 B_e=\begin{pmatrix}1-e&e\\e&1-e\end{pmatrix}.
\end{equation}
Coupling the ideal outcome with an independently corrupted bit string gives
\begin{equation}
 \frac12\norm{Bp-p}_1\le 1-(1-e_{\rm ro})^n\le ne_{\rm ro}.
 \label{eq:readout-bound-ru}
\end{equation}
A preparation error changes the initial state and is propagated through the
visible restriction of the generator. A classical readout error acts on the
final probability vector according to formula \eqref{eq:readout-map-ru}.

\section{Computational verification and numerical applications}

The numerical checks use $\alpha=0.00265$, $\delta=0.001$,
$\omega=0.375$ and dimensionless evolution time $t=1$. These are reference
model parameters.
The checks independently construct the full sparse GKLS superoperator, the orthonormal Krylov basis
$\Span\{G^rE_k\}$ and the analytic kernel
\eqref{eq:damped-fejer-ru}. For $M=4,8,16,32,64$ the full, reduced and analytic
models agree with a maximal error of order $10^{-14}$, and the Krylov
algorithm recovers $\dim\cK_{\rm QPE}=2M-1$ for the nonzero phase shift used
in this comparison. Independent tests of diagonal
Hamiltonian generators for $3\le M\le7$ confirm formula
\eqref{eq:general-diagonal-visible-space-ru}, including intermediate
degeneracies and the upper bound $M^2-M+1$. The three degeneracy regimes of
Proposition~\ref{prop:gkls-toeplitz-ru} are checked at $M=4,8,16,32,64$, and
the two channels of Proposition~\ref{prop:ramsey-separation-ru} reproduce the
distributions \eqref{eq:ramsey-separated-qpe-ru}. These comparisons verify
the consistency of the three numerical representations of the model.

\subsection{Stability beyond the Toeplitz model}

To check Theorem~\ref{thm:residual-stability-ru} consider
$M\in\{8,16,32\}$ and keep the trial space
$\cK_0=\mathcal T_M$ of dimension $2M-1$ built for the generator
\eqref{eq:brownian-gkls-ru}. A GKLS perturbation is added to it:
\begin{equation}
 \cL_\eta=\cL_0-i\eta[H_1,\,\cdot\,]+\eta\mathcal D_{J_1},
 \qquad
 \mathcal D_J(X)=JXJ^\dagger-\frac12\{J^\dagger J,X\},
 \label{eq:non-toeplitz-perturbation-ru}
\end{equation}
where
\begin{align}
 H_1&=c_H\sum_{j=0}^{M-2}
 \left(1+0.35\cos\frac{2\pi(j+1/2)}M\right)
 \bigl(\ket j\bra{j+1}+\ket{j+1}\bra j\bigr),\\
 J_1&=c_J\sum_{j=1}^{M-1}
 \sqrt{0.35+0.65\frac{j}{M-1}}\,\ket{j-1}\bra j.
 \label{eq:non-toeplitz-operators-ru}
\end{align}
Here $\eta\ge0$ so that the added dissipator has nonnegative intensity.
The constants $c_H,c_J$ are chosen from the conditions
$\norm{H_1}_{2\to2}=\norm{J_1}_{2\to2}=1$. Inhomogeneous mixing of
neighbouring branches and amplitude relaxation destroy Toeplitz invariance.

The full dynamics is computed by exponentiating the $M^2\times M^2$
superoperator, while the approximate model uses $A_\eta=PG_\eta P$,
where $G_\eta=\cL_\eta^\dagger$, on the fixed
space $\cK_0$. In Table~\ref{tab:stability-ru} the quantity
$\varepsilon_p=\max_k|p_k-\widetilde p_k|$ is compared with the right-hand side
of \eqref{eq:residual-bound-ru} at $t=1$.

\begin{table}[ht]
\centering
\caption{Non-Toeplitz perturbation: actual error and a posteriori estimate.}
\label{tab:stability-ru}
\begin{tabular}{rrrrr}
\toprule
$M$ & $\dim\cK_0$ & $\norm{R_\eta}_{2\to2}$ & $\varepsilon_p$ & estimate
\eqref{eq:residual-bound-ru}\\
\midrule
$8$  & $15$ & $9.72\times10^{-2}$ & $1.55\times10^{-3}$ & $9.78\times10^{-2}$\\
$16$ & $31$ & $8.99\times10^{-2}$ & $8.02\times10^{-4}$ & $9.03\times10^{-2}$\\
$32$ & $63$ & $8.40\times10^{-2}$ & $4.53\times10^{-4}$ & $8.44\times10^{-2}$\\
\bottomrule
\end{tabular}
\end{table}

The table corresponds to $\eta=10^{-1}$. Automatic exact closure gives the
full dimensions $64$ and $256$ for $M=8$ and $16$. For $M=32$ the full
$1024\times1024$ superoperator is compared with a reduction of dimension $63$.
In all three cases the probability error stays below the a posteriori bound
over the range $10^{-5}\le\eta\le10^{-1}$ and decreases with $M$ for the
chosen normalization of the perturbation. The base generator commutes with the
orthogonal projection onto $\mathcal T_M$, and the pure clock state belongs to
$\mathcal T_M$. Therefore Proposition~\ref{prop:quadratic-residual-ru} gives
the special guarantee $\varepsilon_p=O(\eta^2)$ alongside the uniform guarantee
$O(|\eta|)$ of Corollary~\ref{cor:small-perturbation-ru}. A linear regression
of $\log\varepsilon_p$ on $\log\eta$ over the range
$10^{-5}\le\eta\le10^{-2}$ gives an exponent of approximately $2$
for each of $M=8,16,32$. This is a deterministic check of the scaling over the
indicated range of perturbations.

The margin of the universal bound in the table is approximately $63$, $113$
and $186$ times. Here the $O(|\eta|)$ guarantee for an arbitrary initial state
of a given purity is compared with the $O(\eta^2)$ error for a compatible
state. Therefore the ratio of the bound to the error can grow
as $1/|\eta|$.
At $t=1$ the semigroup factor used is close to unity,
and the numerical bound is close to $\norm{R_\eta}_{2\to2}$.
For the given state the additional first-order cancellation is accounted for by
the special estimate \eqref{eq:quadratic-residual-bound-ru}.

\begin{figure}[ht]
\centering
\includegraphics[width=\textwidth]{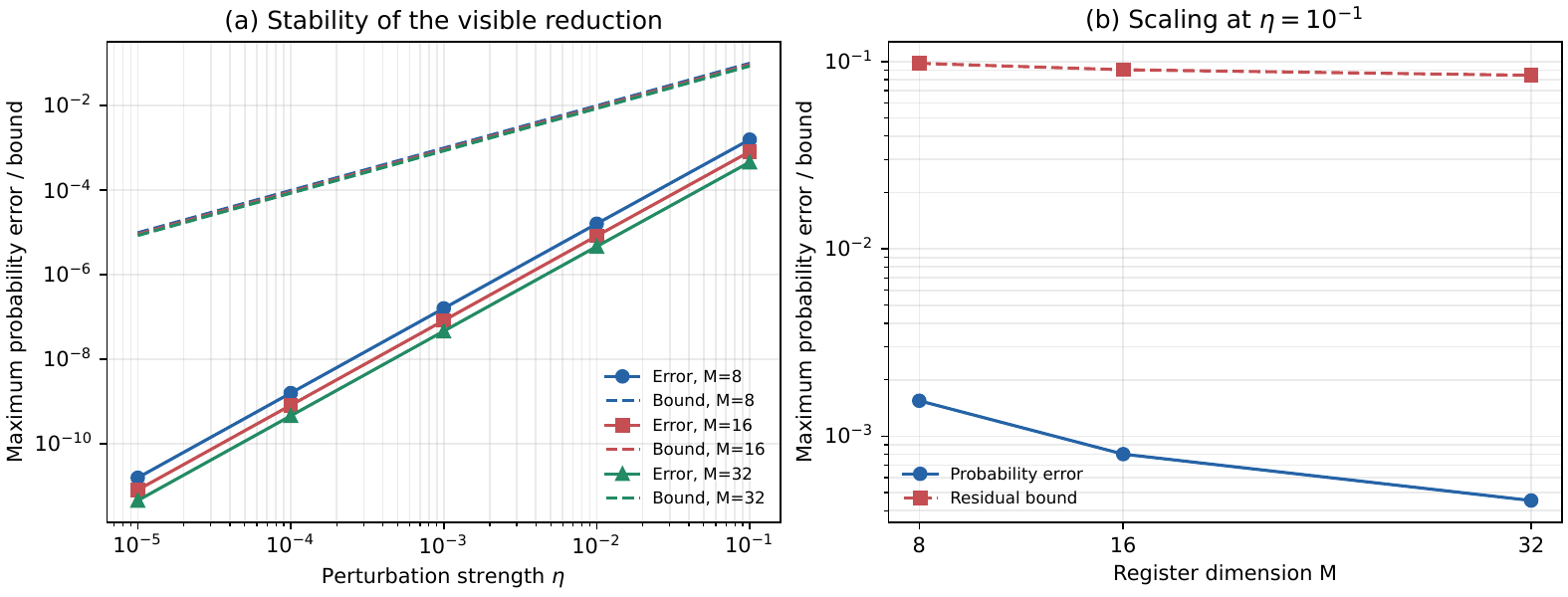}
\caption{Stability of the visible reduction under a non-Toeplitz perturbation
of the generator of strength $\eta$. (a)~Maximal error of the outcome
probabilities (solid lines) and the a posteriori bound (dashed lines) against
$\eta$ on logarithmic scales, for counting registers of $M=8$ (circles), $16$
(squares) and $32$ (triangles) levels. The error decreases quadratically with
$\eta$, whereas the bound is linear. (b)~The same two quantities at
$\eta=10^{-1}$ against the register size $M$: probability error (circles,
solid line) and bound (squares, dashed line). The perturbation destroys the
exact Toeplitz reduction, and the norm of the residual controls its accuracy.}
\label{fig:stability-ru}
\end{figure}

\section{Discrete and time-dependent dynamics}

For a sequence of quantum channels $\Phi_1,\ldots,\Phi_R$ the discrete analogue
of \eqref{eq:visible-space-ru} is the backward space
\begin{equation}
 \cK_{\rm circ}:=\Span\left\{
 E_k,\Phi_R^\dagger(E_k),
 \Phi_{R-1}^\dagger\Phi_R^\dagger(E_k),\ldots,
 \Phi_1^\dagger\cdots\Phi_R^\dagger(E_k)
 \right\}_k.
 \label{eq:discrete-visible-space-ru}
\end{equation}
This construction is consistent with the reduction of discrete and controlled
quantum models \cite{GrigolettoTicozzi2025,GrigolettoViolaTicozzi2025}. For a
time-dependent generator the visible space is generated by solutions of
$\dot E_k(t)=-\cL(t)^\dagger E_k(t)$ with the final conditions $E_k(T)=E_k$.

\section{Conclusion}

The minimal space of QPE-visible observables is the Krylov span of the effects
of the full POVM with respect to the adjoint generator. Universal minimality
characterizes it as the smallest exact linear model of all outcome
probabilities. For QPE the measurement projection is computed through cyclic
sums of diagonals. If the generator is diagonal on the matrix units, then
$\dim\cK_{\rm QPE}=\sum_r|\Lambda_r|$ and
$M\le\dim\cK_{\rm QPE}\le M^2-M+1$; the basis is given by the spectral classes
inside the residues $j-\ell\bmod M$.

The finite-time Gramian unites the algebraic rank and the approximate
visibility: its range coincides with $\cK_{\rm QPE}$, the leading eigenspaces
give the optimal spectral truncation, and the weight of the perturbed modes
outside the exact unperturbed visible space is of order $O(\eta^2)$.
The a posteriori estimate in terms of the residual $(I-P)G P$ controls each
probability for an approximately invariant reduction. Two random-unitary models
with a common Ramsey contrast and different QPE distributions show that the
measurement distinguishes cyclic combinations of higher harmonics.

The explicit GKLS realization gives the damped Fej\'er kernel and the exact
dimension regimes $M$, $2M-2$ and $2M-1$. The full propagator, the Krylov
reduction and the analytic kernel agree in the calculations to within an
accuracy of order $10^{-14}$. The non-Toeplitz GKLS perturbation demonstrates
compliance with the a posteriori bound and an approximately quadratic error
order at $M=8,16,32$. The QPE-visible modes determine the output distribution,
and the residual estimate bounds the error of a reduction built from them.

\section*{Competing interests}
The authors declare no conflicts of interest.

\section*{Funding}
This research received no external funding.

\section*{Author contributions}
Marina~O.~Gordeichuk: Conceptualization, Software, Validation, Visualization, Writing -- review \& editing.
Oleg~M.~Kiselev: Conceptualization, Methodology, Formal analysis, Supervision, Writing -- original draft.

\section*{Data and code availability}
The numerical scripts, result tables and instructions for reproducing the
checks are included in the ancillary files accompanying this preprint.
The scripts compare the full GKLS dynamics, Krylov reduction and analytic
kernel, check spectral degeneracies and the Ramsey counterexample, and
evaluate the non-Toeplitz perturbation experiment. These are deterministic
numerical checks of the stated mathematical models.

\begingroup
\raggedright

\endgroup

\end{document}